\documentclass[10pt]{sig-alternate-05-2015}

\usepackage{amsmath, cuted}
\usepackage{subcaption}

\newtheorem{Definition}{Definition}

\newtheorem{theorem}{Theorem}[section]
\newtheorem{lemma}[theorem]{Lemma}

\newtheorem{example}[theorem]{Example}

\newtheorem{remark}[theorem]{Remark}
\newtheorem{fact}[theorem]{Fact}

\usepackage{paralist}
\usepackage{amsfonts}
\usepackage{amssymb}
\usepackage{tikz}

\newcommand{\equalityref}[1]{\hyperref[#1]{Equality~\eqref{#1}}}
\newcommand{\inequalityref}[1]{\hyperref[#1]{Inequality~\eqref{#1}}}

\newcommand{\N}{\mathbb{N}}

\newcommand{\Oh}{\mathcal{O}}

\DeclareMathOperator{\In}{In}
\DeclareMathOperator{\Out}{Out}
\newcommand{\G}{\mathcal{G}}

\usepackage{cleveref}

\title{}

\begin{document}
  \title{Linear-Time Data Dissemination in Dynamic Networks
\titlenote{This work has been supported by the Austrian Science Fund (FWF)
projects ADynNet (P28182) and RiSE (S11405). }
}

\numberofauthors{3}
\author{
\alignauthor
Manfred Schwarz \\ 
       \affaddr{Institute of Computer Engineering}\\
       \affaddr{TU Wien}\\
       \email{mschwarz@ecs.tuwien.ac.at}
\alignauthor
Martin Zeiner \\ 
       \affaddr{Institute of Computer Engineering}\\
       \affaddr{TU Wien}\\
       \email{mzeiner@ecs.tuwien.ac.at}
\alignauthor Ulrich Schmid \\ 
       \affaddr{Institute of Computer Engineering}\\
       \affaddr{TU Wien}\\
       \email{s@ecs.tuwien.ac.at}
}

  \date{\today}
  \maketitle
\begin{abstract}
Broadcasting and convergecasting are pivotal services in distributed
systems, in particular, in wireless ad-hoc and sensor networks, which
are characterized by time-varying communication graphs. We
study the question of whether it is possible to disseminate data
available locally at some process to all $n$ processes in sparsely 
connected synchronous dynamic networks with directed links 
in linear time. Recently, Charron-Bost, F\"ugger and Nowak proved 
an upper bound of $\Oh (n \log n)$ rounds for the case where every
communication graph is an arbitrary directed rooted tree. We present a new
formalism, which not only facilitates a concise proof of this result,
but also allows us to prove that $\Oh(n)$ data dissemination is
possible when the number of leaves of the rooted trees are bounded
by a constant. In the special case of rooted chains, only $(n-1)$ 
rounds are needed. Our approach can also be adapted for undirected
networks, where only $(n-1)/2$ rounds in the case of arbitrary
chain graphs are needed. 
\end{abstract}


\section{Introduction}

We consider a synchronous network of $n$ failure-free nodes with unique 
ids (uids), which 
are connected by directed point-to-point links. The nodes execute a
deterministic algorithm for disseminating some local data (say, 
the uids for simplicity), which shall ensure
that the uid of at least one node becomes known to all nodes in the 
system as fast as possible. In a synchronous distributed system, the execution proceeds in the form of
lock-step rounds $r=1,2,\dots$, where all processes send and receive
round-$r$ messages and simultaneously execute a computing step, which
also starts the next round. 
Communication is unreliable,
though: A \emph{message adversary} \cite{AG13} determines 
which receiver gets a
message from the respective sender in a round: It effectively generates
a sequence $G_1,G_2,\dots$ of directed \emph{communication graphs},
where $G_r$ contains a directed 
edge $(p,q)$ if the message from $p$ is received by $q$ in round $r$.
We assume that the set of nodes and hence $n$ is fixed, whereas
the edges may change over time. Messages may have arbitrary 
size, i.e., we adhere to the LOCAL model \cite{Pel00}.

The particular question asked in this paper is: How many rounds are needed until
the uid of some node is known to all $n$ nodes, for a certain message adversary? We will call this quantity the \emph{dissemination time},
and it is obvious that small dissemination times 
are beneficial for data distribution applications.\footnote{It may be argued that adhering to the LOCAL model is overly simplistic for e.g.\ wireless settings, where the CONGEST model 
(a maximum message size of $O(\log n)$) is usually considered more 
appropriate. However, in the light of our past efforts to solve 
the problem with arbitrary message size, we consider is close to hopeless 
to immediately address the data dissemination in the CONGEST model. A solution
in the LOCAL model, however, might pave the way to the latter also, as it might
suggest ideas for how to properly schedule the processes' activities in
order to avoid congestion, cp.\ \cite{HW12:PODC}.} This even includes consensus
algorithms like \cite{BRS12:sirocco,CG13,SWSBR15:NETYS,SWS16:ICDCN}, since system-wide 
agreement obviously requires system-wide data dissemination. Moreover, 
small dissemination times are also interesting for data aggregation, which 
is a pivotal task in 
wireless sensor networks \cite{ASSC02}. After all, convergecasting is the dual of
broadcasting: By reverting the direction of the links and the sequence
of communication graphs, a successful broadcast becomes a successful 
convergecast.

The dissemination time obviously depends heavily upon the message adversary,
i.e., the actual sequence of communication graphs $G_1,G_2,\dots$: If
e.g.\ $G_1$ contains a star, it is 1, if every graph consists of the same two weakly
connected components, it is $\infty$. We are interested in an upper bound on
the dissemination time for at least sparsely connected communication graphs. 
More specifically, we restrict our attention to the case of an \emph{oblivious}
message adversary \cite{AG13,CG13}, where $G_1,G_2,\dots$ is an arbitrary sequence 
of graphs each drawn from a set ${\cal G}$ of candidate graphs with each 
$G\in {\cal G}$ containing some rooted spanning tree. Note that 
this is actually the weakest per-graph 
restriction that guarantees a finite worst-case dissemination time for oblivious
message adversaries. 

A relatively simple pigeonhole argument (see Lemma~\ref{lem:infl})
yields an upper bound of $\Oh(n^2)$ for the dissemination time in this case.\footnote{For all
the $\Oh(.)$ terms in this paper, the constants can be computed.}
Recently, Charron-Bost, F\"ugger and Nowak improved this to
$\Oh (n \log n)$. We conjecture that this bound can be further
tightened to linear time $\Oh(n)$. 
Albeit we were not yet able to prove or disprove this, despite considerable efforts, 
we establish results in this paper
that back-up our conjecture.

\textbf{Main contributions and paper organization:} After a short discussion of 
related work in Section~\ref{sec:relwork} and a description of our system model 
in Section~\ref{sec:model}, we provide the following results in Section~\ref{sec:dir}:
\begin{compactenum}
\item[(i)] We introduce the concept of influence and covering sets and apply these 
techniques in a novel and very concise proof of the known $\Oh (n\log n)$ upper bound
for arbitrary directed rooted trees.
\item[(ii)] We show that a dissemination time of $\Oh(n)$ can be guaranteed for
directed rooted trees with a constant number of leaves. In the case of
directed rooted paths, i.e., directed rooted trees with only 
one leaf, the dissemination time is only $(n-1)$ rounds.
\item[(iii)] In Section~\ref{sec:undir}, we adapt our approach to undirected
networks and show that only $(n-1)/2$ rounds are needed for the dissemination time 
in the case of arbitrary (undirected) chain graphs. 
\end{compactenum}
Some outlook in Section~\ref{sec:outlook} concludes our paper.

\section{Related Work}\label{sec:relwork}

Research on broadcasting and gossiping has a long history
(see~\cite{HHL88} for a survey) and many variations of these problems have
been studied: For the classical telephone problem we refer to~\cite{FHMP79,
	Wes82,PPS15}, and for radio broadcasting to~\cite{EG06, EGS08} and
references therein. In~\cite{DHSZ06} the authors consider the
rendezvous-communication-model.
The widely-used push/pull/push-pull models~\cite{ES09a} were studied on
several graph classes like the complete graph~\cite{FG85, Pit87}, hypercubes
and Erd\H{o}s-Renyi-graphs~\cite{FPRU90, FM94} random geometric
graphs~\cite{FSS13}, and preferential-attachment graphs~\cite{DFF11}. A
list-based quasi-randomized algorithm has been studied in~\cite{BFHV12,DHL13}.
They have in common that they focus on broadcasting a message from a fixed
source on a (possibly random but) fixed graph to all nodes.

Radio broadcasting in evolving graphs has been investigated in
\cite{CMPS09,KLNOR10,AGKM16}, \cite{DPRSV13} focused on token-forwarding
algorithms, and \cite{GSS14,CCDFPS16} studied push/pull-algorithms
on dynamic graphs. Flooding algorithms have also been studied under
several models, including edge-Markovian and related models
\cite{BCF09,CMMPS10,CST15}.
In the edge-Markovian model an edge present at time $t$ stays present with
probability $p$ and disappears with probability $1-p$ at time $t+1$ and an
absent edge appears with probability $q$ and remains absent with $1-q$. 
Whereas those papers focus on the broadcasting of a single item from a fixed
source to all nodes,~\cite{KLO10, BCF12} also consider $k$-token dissemination
and all-to-all algorithms. All this work above considers undirected communication
graphs, however.

\section{Model}\label{sec:model}
We consider a set of processes $\Pi = \{1,\dots ,n\}$ with uids, connected
by directed point-to-point links. The processes execute a deterministic
full-information protocol for distributing a unique local piece of 
information (for ease of exposition, the uid) to the other processes.
The distributed computation proceeds in an infinite number 
of synchronous lock-step
rounds $r=1,2,\dots$. Each round $r$ consists of a communication-closed 
message exchange, specified by the communication graph $G_r$ determined
by an oblivious message adversary \cite{AG13}, followed
by a simultaneous local computing step at every process.
In a full information protocol, every process sends its complete 
state in every round: If process $p$ receives the state of some different
process $q$ (reached at the end of round $r>0$ resp.\ the initial
state for $r=0$) in round $r+1$, it 
forwards $q$'s state as part of its own state in all 
following rounds.

An execution of our system is just an infinite sequence of 
rounds. It can be uniquely
described by an initial configuration $C_0$, which is
the vector of the initial states (that includes the uid)
of every process, followed by an infinite sequence $\G$ of 
communication graphs $G_1,G_2,,\dots$. The configuration
reached at the end of round $r$, after the computation step,
is denoted by $C_r$.

%
%
Formally, let $G$ be a directed or undirected labeled graph on $n$ vertices and let $\mathcal{G} = (G_r)_{r=1}^{\infty}$ be an infinite sequence of such graphs.
Moreover, let $\sigma_i$ be the finite prefix of $\mathcal{G}$ of length $i$;
we may drop the index if the length is clear from the context.
Let $\In_p(r)$ resp.\  $\Out_p(r)$ denote the set of incoming 
resp.\ outgoing edges of node $p$ in $G_r$. 

Thanks to our full information protocol, every node has knowledge 
$K_p(r)$ at the end of round $r$ (with $K_p(0)$ representing the initial 
knowledge), which adheres to the following rules:
\begin{compactenum}[(1)]
 \item Initial state: $K_p(0) = \{p\}$ for all $p \in \Pi$ (every node knows only its own uid at the beginning).
 \item Updating rule: The knowledge $K_p(r+1)$ process $p$ obtains
at the end of round $r+1$ is its previous knowledge $K_p(r)$ together 
with the information he gets via all incoming edges in $G_{r+1}$, i.e., 
\[
 K_p(r+1) = K_p(r) \cup \bigcup_{q: \ (q,p) \in \In_p (r+1)}  K_q(r) .
\]
\end{compactenum}
Subsequently, we will use phrases like ``$K_p(t)$ at time $t$'' for
integer times $t\geq 0$, which means $K_p(r)$ at (the end of) 
round $r=t$ for $t>0$ and $K_p(0)$ otherwise.

The dissemination time, given a sequence of graphs, is the first time all 
processes learned the uid from a common process, which is formally
defined as follows:
\begin{Definition}
Given a class of graphs $\mathfrak{G}$ on $n$ nodes and an infinite sequence
of graphs $\G\in \mathfrak{G}^\N$, i.e., $G_r \in \mathfrak{G}$ for
$r\geq 1$, the dissemination time in $\G$ is defined as
\begin{equation}\label{equ:brseq}
 B_{\mathcal{G}} = \min \left\{ r: \ \bigcap_{p \in \Pi} K_p(r) \neq \emptyset \right\}.
\end{equation}
The dissemination time of the class $\mathfrak{G}$ is defined as
\begin{equation}\label{equ:brcla}
 B^{\mathfrak{G}} = \max_{\mathcal{G} \in \mathfrak{G}^{\N}} B_{\mathcal{G}}.
\end{equation}
\end{Definition}

In this paper, we restrict our attention to classes of graphs 
where every element is\footnote{Actually, we can immediately generalize
all our results to the case where every element only \emph{contains} 
a rooted tree etc., as additional edges can only speed-up data dissemination.} 
a rooted tree ($\mathfrak{T}$), its subclass consisting of
directed chains ($\mathfrak{C}$), as well as their undirected analoga.

\begin{remark}
Following a commonly used definition \cite{FHMP79, Lab89, Ger94}, we could 
also define a dissemination time $B_{\mathcal{G}}(p) = \min \left\{ r: \ 
\bigcap_{q \in \Pi} K_q(r) = \{p\} \right\}$ for node $p$ in the first place, i.e.,
the time when $p$ becomes known to all nodes in a graph sequence $\G$.
The dissemination time \eqref{equ:brseq} is then 
$B_{\mathcal{G}}= \min_{p \in\Pi} B_{\mathcal{G}}(p)$. Whereas this alternative
expression has been used rarely for dynamic graphs, it has been employed
for static graphs~\cite{Lab89, FHMMM94}; an analogon of~\eqref{equ:brcla} 
has also been studied in~\cite{Lab89}. On the other hand, in the classic 
telephone problems and its variants, one is interested in 
$\max_{p \in\Pi} B_{\mathcal{G}}(p)$.
\end{remark}

\begin{example} To illustrate the definitions above we give a short example (see Figure~\ref{fig:example}).
	$~~$
\newcommand{\impscale}{0.9}
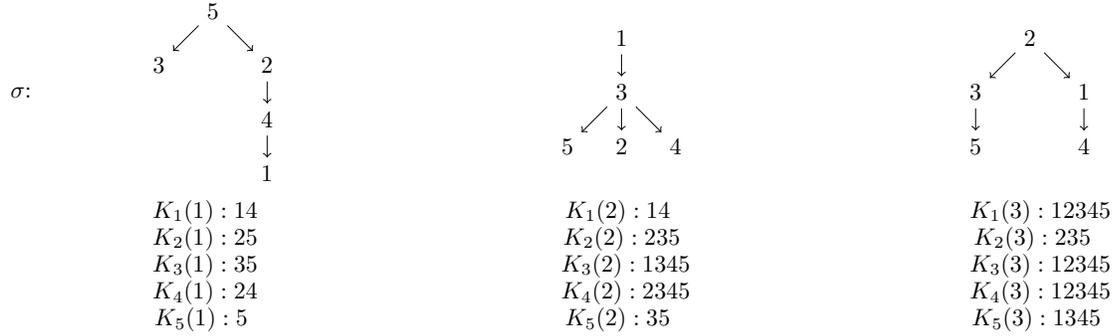
\begin{figure*}[!ht]
 \centering
   \begin{subfigure}{0.10\linewidth}
     \scalebox{\impscale}{
       \begin{tikzpicture}
         \node at (0, 0.5) {$\sigma$:};
       \end{tikzpicture}
     }
   \end{subfigure}
   \begin{subfigure}{.30\linewidth}
     \scalebox{\impscale}{
       \begin{tikzpicture}
         \node (p5) at (0,0) {$5$};
         \node (p2) at (0.8,-0.8) {$2$};
         \node (p3) at (-0.8,-0.8) {$3$};
         \node (p4) at (0.8,-1.6) {$4$};
         \node (p1) at (0.8,-2.4) {$1$};

		  \draw[->] (p5) -- (p3);
		  \draw[->] (p5) -- (p2);
		  \draw[->] (p2) -- (p4);
		  \draw[->] (p4) -- (p1);
       \end{tikzpicture}
     }
   \end{subfigure}
   \begin{subfigure}{.30\linewidth}
     \scalebox{\impscale}{
       \begin{tikzpicture}
         \node (p1) at (0,0) {$1$};
         \node (p3) at (0.0,-0.8) {$3$};
         \node (p5) at (-0.8,-1.6) {$5$};
         \node (p4) at (0.8,-1.6) {$4$};
         \node (p2) at (0.0,-1.6) {$2$};

		  \draw[->] (p1) -- (p3);
		  \draw[->] (p3) -- (p2);
		  \draw[->] (p3) -- (p4);
		  \draw[->] (p3) -- (p5);
       \end{tikzpicture}
     }
   \end{subfigure}
   \begin{subfigure}{.25\linewidth}
     \scalebox{\impscale}{
       \begin{tikzpicture}
         \node (p2) at (0,0) {$2$};
         \node (p1) at (0.8,-0.8) {$1$};
         \node (p3) at (-0.8,-0.8) {$3$};
         \node (p4) at (0.8,-1.6) {$4$};
         \node (p5) at (-0.8,-1.6) {$5$};

		  \draw[->] (p2) -- (p3);
		  \draw[->] (p2) -- (p1);
		  \draw[->] (p1) -- (p4);
		  \draw[->] (p3) -- (p5);
       \end{tikzpicture}
     }
   \end{subfigure}\\
   \begin{subfigure}{.10\linewidth}
     \scalebox{\impscale}{
       \begin{tikzpicture}
       \end{tikzpicture}
     }
   \end{subfigure}
   \begin{subfigure}{.30\linewidth}
     \scalebox{\impscale}{
       \begin{tikzpicture}
         \node (p1) at (0,0)    {$K_1(1):14~~~$};
         \node (p2) at (0,-0.4) {$K_2(1):25~~~$};
         \node (p3) at (0,-0.8) {$K_3(1):35~~~$};
         \node (p4) at (0,-1.2) {$K_4(1):24~~~$};
         \node (p5) at (0,-1.6) {$K_5(1):5~~~~$};
       \end{tikzpicture}
     }
   \end{subfigure}
   \begin{subfigure}{.30\linewidth}
     \scalebox{\impscale}{
       \begin{tikzpicture}
         \node (p1) at (0,0)    {$K_1(2):14~~~$};
         \node (p2) at (0,-0.4) {$K_2(2):235~~$};
         \node (p3) at (0,-0.8) {$K_3(2):1345~$};
         \node (p4) at (0,-1.2) {$K_4(2):2345~$};
         \node (p5) at (0,-1.6) {$K_5(2):35~~~$};
       \end{tikzpicture}
     }
   \end{subfigure}
   \begin{subfigure}{.25\linewidth}
     \scalebox{\impscale}{
       \begin{tikzpicture}
         \node (p1) at (0,0)    {$K_1(3):12345$};
         \node (p2) at (0,-0.4) {$K_2(3):235~~$};
         \node (p3) at (0,-0.8) {$K_3(3):12345$};
         \node (p4) at (0,-1.2) {$K_4(3):12345$};
         \node (p5) at (0,-1.6) {$K_5(3):1345~$};
       \end{tikzpicture}
     }
   \end{subfigure}
 \caption{example execution with knowledge sets}\label{fig:example}
\end{figure*}

The dissemination times are $B_{\mathcal{G}}=B_{\mathcal{G}}(3)=B_{\mathcal{G}}(5) = 3$.
\end{example}

\section{Rooted Trees}\label{sec:dir}

In this section, we will present our main results on rooted trees
$\mathfrak{T}$.
We use influence sets (see~\cite[Lemma 3.2. (b)]{KLO10}) for
this purpose, which are are dual to knowledge sets:
While the knowledge set $K_p(r)$ describes which processes node $p$ has already
heard of at the end of round $r$, the influence set $S_p(r)$ describes which 
processes have already heard of $p$:
\begin{Definition}
The influence set $S_p (r)$ of process $p$ at time $r$ is
the set of processes that know about $p$ at time $r$, i.e.,
$S_p (r) = \{ q \in \Pi : \ p \in K_q (r) \}$.
\end{Definition}

\subsection{Influence Sets and Coverings}

Obviously, there are always $n$ influence sets, and each node can be element
of multiple of these. In the following lemma, we collect some
elementary properties of influence sets.

\begin{lemma}\label{lem:infl} For all $p \in \Pi$ and $r \geq 0$, we have:
 \begin{compactenum}[(i)]
  \item Initial state: $S_p(0)= \{p\}$.
  \item Updating rule: $S_p(r+1) = S_p (r) \cup \bigcup_{q \in S_p(r)} \{ q': \ (q,q') \in \Out_q(r+1) \}$.
  \item Given $\G$, the dissemination time $B_{\mathcal{G}} = \min \{ r: \ \max_p |S_p(r)| = n \}$.
  \item $S_p(r) \subseteq S_p(r+1)$ for all $p$ and all $r$.
  \item\label{lem:infl:rootgrows} If $p$ is the root in $G_{r+1} \in \mathcal{G}$, $p\in S_q(r)$, and $|S_q(r)|<n$ then $|S_q (r+1)| > |S_q(r)|$.
 \end{compactenum}
\end{lemma}
\begin{proof}
Properties (i) -- (iv) are an immediate consequence of the definition. Let us prove (v): 
Define $X:=\Pi \setminus S_q(r)$ as the nonempty set of nodes which don't know from $q$. Since $p$ is the root, for all $v \in \Pi$ (and in particular for all $v\in X$) there exists
a path from $p$ to $v$  in $G_r$. Thus there must be an edge from $S_q(r)$ to $X$ and hence $S_q(r)$ grows at least by $1$.
\end{proof}

From Property~(v) of this lemma, we obtain directly the trivial $\Oh (n^2)$-bound on the dissemination time for rooted trees $\mathfrak{T}$:
By the pigeonhole principle, after $n(n-2)+1$ rounds, one node was at least $n-1$ times the root and hence its influence set has size $n$. 

Whereas influence sets will turn out to be sufficient for establishing 
our results on chain graphs, 
we need the extended concept of coverings for dealing with general rooted trees.

\begin{Definition} For $r\geq t$ let $$\mathcal{C}_{I(t)}(r) = \{ S_p(r) \ | \ p \in I(t) \}$$
be a class of influence sets, for some given index set $I(t) \subseteq \Pi$.
It is called covering if
$\bigcup_{S_p(r) \in \mathcal{C}_{I(t)}(r)} S_p(r) = \Pi$.
The influence sets that make up a covering are called covering sets. 
The size of a covering is the number of covering sets it consists of, i.e., the size of its index set $I(t)$.
(For sake of simplicity we will drop the index sometimes in the following if $r=t$ and the index set is clear from the context.)


A sequence of coverings $(\mathcal{C}_{I(r)}(r))_{r \geq 0}$ with the additional property $I(r+1) \subseteq I(r)$ we denote by $\mathfrak{C}$.
\end{Definition}

Clearly, a trivial example of a covering is the set of all influence sets. To exclude such
trivial cases, we introduce a subclass of coverings that contain no redundant sets.

\begin{Definition}
A strict covering $\mathcal{SC}_{I(r)}(r)$ is a covering with the property that $\mathcal{SC}(r) \setminus \{S\}$ $\forall S \in \mathcal{SC}(r)$ is not a covering.
 A unique node is a node that is element of only one covering set.
\end{Definition}

The following lemma states some useful properties of coverings:

\begin{lemma}\label{lem:cov} In the case of rooted trees 
$\mathfrak{T}$, every covering satisfies the following properties:
\begin{compactenum}[(i)]
  \item A covering is a strict covering iff each covering set contains a unique node.
  \item\label{lem:cov:ii} In a strict covering every covering set, except possibly one, loses at least one of its unique nodes. 
	The only covering set that may not lose one of its unique nodes is the
	one of the root of $G_r$.
  \item\label{lem:cov:red} Let $\mathcal{C}_{I(r)}(r)$ be a strict covering at time $r$.  
Assume that at time $r+t$, for some $t>0$, there exists a covering set $S$ in $X:=\mathcal{C}_{I(r)}(r+t)$ containing no unique node.
Then, $X':=X\setminus \{S\}$ is still a covering and $X'$ has at most $|S|$ more unique nodes than $X$.
By repeating this argument,
one can reduce $X$ to a new strict
covering $\mathcal{SC}_{I(r+t)}(r+t)$ with strictly smaller 
index set $I(r+t) \subset I(r)$.
 \item Let $\ell$ be the number of leaves in $G_r$. In a strict covering $\mathcal{SC}_{I(r-1)}(r-1)$, at most $\ell$ influence sets do not grow in round~$r$.
 \item If a strict covering consists of only one set, then dissemination has been completed.
 \item\label{lem:cov:ZA} At time $t=2$, there is always a covering consisting of covering sets of size at least $2$. 
  \item\label{lem:cov:onesetgrows} For each covering $\mathcal{C}_{I(r)}(r)$ there exists a $p\in I(r)$ with $|S_p(r+1)|>|S_p(r)|$.
\end{compactenum}
\end{lemma}

\begin{proof}

The properties (i) and (v) are obvious.

(ii) Let $q$ be a unique node. Then $q$ is a) the root, or b) has a predecessor which is only in the same influence as $q$ (and thus unique too), or c) has a predecessor which is in another
influence set too. In case c) $q$ is not unique anymore. In case b) we repeat the argument with the predecessors of $q$ until we stop in case c) or a). 
If we stop in case c), then one of the predecessor of $q$ is not unique
anymore. If we stop in case a) (and  not in c)!) then $q$ and all its predecessors remain unique. Thus the only influence set which may not lose one of its
unique nodes is the set of the root.

(iii) Since $S$ contains no unique nodes all these nodes most be contained in other sets too, and thus we can remove $S$ from $X$ still having a covering.
Repeating this procedure leads to a covering $X'$ containing no unique nodes. Hence it is a strict covering.

(iv) Note that an influence $S$ which does not grow has the property that for all $p \in S$ also all successors of $p$ must be contained in $S$. Moreover, each set in $\mathcal{SC}(r-1)$ contains
a unique node. Now let $P_i$ denote the unique path from the root to the leaf $l_i$ ($1\leq i \leq \ell$). We will show that each path $P_i$ contains unique nodes from at most one
non-growing influence set. Let $p$ be a unique node of a non-growing influence set $S$. Then $S$ contain all successors of $p$ and those nodes can not be unique nodes from another
non-growing set. On the other hand, if there is a unique node $p'$ from another non-growing set $S'$ on the path from the root to $p$ then $S'$ would contain all successors from $p'$, hence
also $p$, and $p$ would not be a unique node. Consequently there can be at most $\ell$ non-growing covering sets.

(vi) The only influence sets of size $1$ after round $1$ are the sets $S_{l_1} (1),\dots ,S_{l_\ell}(1)$ where the $l_i$ are the leaves in the tree $G_1$. But node $l_i$ is surely
contained in the influence set of ist predecessor. So take any covering that does not contains the influence sets of leaves but the influence sets of predecessors of leaves.

(vii) Since $\mathcal{C}(r)$ is a covering, one set must contain the root. By Lemma~\ref{lem:infl}~\eqref{lem:infl:rootgrows}, this set grows.
\end{proof}

\subsection{Bounds on dissemination time}

Equipped with the properties from Lemma~\ref{lem:cov}, we can now give a novel, concise proof of the
$\Oh(n\log n)$-bound established in~\cite{CFN15:ICALP,CBS06}.
\begin{fact}[{\cite[Lemma 4]{CFN15:ICALP} and~\cite[Lemma 1]{CBS06}}]
For the class $\mathfrak{T}$ of rooted trees, dissemination is completed within 
$B^{\mathfrak{T}} = \Oh(n\log n)$ rounds.
\end{fact}
\begin{proof}
Let $\mathcal{C}_{I(r)}(r)$ be a strict covering of size $x+1$, and $z$ be the number of unique nodes
in $\mathcal{C}(r)$. Lemma~\ref{lem:cov}~\eqref{lem:cov:ii} ensures that, after $t:=\lceil \tfrac{z}{x} \rceil$ rounds, there
exists an influence set in $\mathcal{C}_{I(r)}(r+t)$ with no unique nodes, which can be removed.
The new strict covering $\mathcal{SC}_{I(r+t)}(r+t)$ resulting from this procedure is of size at most $x$.

Starting at $r=0$ and $x+1=n$ and using $z\leq n$, we can bound the dissemination time by
$
B^{\mathfrak{T}} \leq \sum_{i=1}^{n} \left\lceil \frac{n}{i} \right\rceil = \Oh(n\log n) 
$.
\end{proof}

If we restrict ourselves to rooted trees with 
a fixed number of leaves, 
it is possible to prove that data dissemination can be completed 
even in linear time.

\begin{theorem}\label{thm:kleaves}
For the class of rooted trees $\mathfrak{T}_{k-1}$ with exactly 
$k-1$ leaves, data dissemination is $B^{\mathfrak{T}_{k-1}}\leq k\cdot (n-3)+2$ 
rounds.
\end{theorem}
\begin{proof}
Let $\mathfrak{SC}$ be a sequence of strict coverings with $I(r)\subseteq I(r+1)$ such that at time $r=2$ all influence sets are of size at least $2$ (see Lemma~\ref{lem:cov}~(\ref{lem:cov:ZA})). 
Let $S_{p_1}(r),\dots , S_{p_{k}}(r)$ be the $k$ smallest influence sets in $\mathfrak{SC}$ at time $r$ and
let $s_{p_i}(r)$ denote the size of $S_{p_i}(r)$. Due to Lemma~\ref{lem:cov}~(iv) in every round at least one of them grew by at least one, hence
\[
\sum_{i=1}^{k} s_{p_i}(r) \geq 2k+(r-1).
\]
Thus, if $2k+(r-1) = k(n-1)+1$ one set must contain $n$ elements and dissemination is done. Solving this equation for $r$ yields $r=(k(n-3)+2)$.
\end{proof}

In particular, in the special case of a directed chain graph (i.e., $k=2$), 
one can do even better:
The following theorem shows that the dissemination time is only $n-1$ rounds
in this case.
Since in the constant chain graph,
it takes the root $n-1$ rounds to disseminate its value, this bound is tight.

\begin{theorem}
Let $\mathfrak{C}$ be the class of directed chains.
At the end of round $r$, there exists a collection
$\mathcal{S}(r)=\{S_{p_1}(r),\dots ,S_{p_{n-r}}(r) \}$ of $n-r$ influence sets with 
\begin{compactenum} 
 \item $|S_{p_i}| \geq r+1$ for $1\leq i \leq n-r$, and
 \item for all $0 \leq r \leq n-1$,
\begin{equation}\label{equ:unionsize}
 \left| \bigcup_{i \in I} S_{p_i}(r) \right| \geq r+1 + |I|-1 \quad \forall \ I \subseteq [1,n-r] .
\end{equation}
\end{compactenum}
 Thus, $B^\mathfrak{C} \leq n-1$ rounds.
\end{theorem}

\begin{proof}
We do it by induction on $r$. If $r=0$, then $S_p(0) = \{p\}$ and thus $|S_p(0)|=1$ for all $p$. Obviously, $ \left| \bigcup_{i \in I} S_{p_i}(0) \right| = |I|$. Assume that
the induction hypothesis holds for $r$. We will show the the stated assertion also holds for $(r+1)$.

We take successively a set $S_{p_i}(r+1)$ (where $S_{p_i}(r)$ was contained in $\mathcal{S}(r)$) and add it to $\mathcal{S}(r+1)$ iff the following two conditions hold:
\begin{itemize}
\item[(i)] $\left| S_{p_i}(r+1) \right| \geq r+2$ and
\item[(ii)] inequality~\eqref{equ:unionsize} holds for all collection of sets of $\mathcal{S}(r+1)$ and $S_{p_i}(r+1)$.
\end{itemize}

Note that condition~(i) holds for all but at most one set from $\mathcal{S}(r)$: A set of size $m$ does not grow iff it captures the last $m$ elements of the chain. Since
due to inequality~\eqref{equ:unionsize} all sets of size $(r+1)$ are pairwise different only one of these sets can be completely at the end of the chain, hence at most
one set does not grow. If such a set exists we denote it by $S_{p*}$.

Assume now that we have already added successfully $(L-1)$ sets to $\mathcal{S}(r+1)$ and that by adding $S_{p_L}(r+1)$ condition (ii) is violated, i.e., there exists
$k \leq (L-1)$ sets $S_{p_{i_1}}(r+1),\dots,S_{p_{i_k}}(r+1) \in \mathcal{S}(r+1)$ such that
\[
\left| \bigcup_{\ell = 1}^k S_{p_{i_{\ell}}}(r+1) \right| = r+1+k
\]
and
\[
 \left| \bigcup_{\ell = 1}^k S_{p_{i_{\ell}}}(r+1) \cup S_{p_L}(r+1) \right| = r+1+k.
\]
This means that $S_{p_L}(r+1) \subseteq \bigcup_{\ell = 1}^k S_{p_{i_{\ell}}}(r+1)$.

By induction hypothesis,
\[
\left| \bigcup_{\ell = 1}^k S_{p_{i_{\ell}}}(r) \cup S_{p_L}(r) \right| \geq r+1+k.
\]
Thus the set $\bigcup_{\ell = 1}^k S_{p_{i_{\ell}}}(r) \cup S_{p_L}(r)$ did not grow in round $(r+1)$, or equivalently, it captures the last $(r+1+k)$ elements of
the chain in round $(r+1)$. Firstly, in case of the existence of $S_{p*}$, also this set contains exactly the last $(r+1)$ elements, which gives
\[
\left| \bigcup_{\ell = 1}^k S_{p_{i_{\ell}}}(r) \cup S_{p_L}(r) \cup S_{p*}(r) \right| = r+1+k
\]
which is a contradiction to the induction hypothesis. Hence -- if $S_{p*}$ exists -- we can add all influence sets of $\mathcal{S}(r)$ but $S_{p*}$ to $\mathcal{S}(r+1)$ and
the assertion of this theorem holds in this case. On the other hand we are allowed to delete one set from our 'good' sets, so we do not add $S_{p_L}(r+1)$ to $\mathcal{S}(r+1)$.
The remaining question is: Can there be an index $L'>L$ such that again condition (ii) is violated? So assume that there is such an index and $k_1<L'$ sets
$S_{p_{i_1}},\dots,S_{p_{i_{k_1}}}$ with
\[
\left| \bigcup_{\ell = 1}^{k_1} S_{p_{i_{\ell}}}(r+1) \cup S_{p_{L'}}(r+1) \right| = r+1+k_1.
\]
Again, due to induction hypothesis, the set $\bigcup_{\ell = 1}^{k_1} S_{p_{i_{\ell}}}(r) \cup S_{p_{L'}}(r)$ captures the last $(r+1+k_1)$ elements of the chain
in round $(r+1)$. If $k_1\geq k$ then we take the $k_1$ influence sets from here together with $S_{p_L}(r)$ and $S_{p_{L'}}(r)$ and obtain
\[
\left| \bigcup_{\ell = 1}^{k_1} S_{p_{i_{\ell}}}(r) \cup S_{p_{L'}}(r) \cup S_{p_{L}}(r) \right| = r+1+k_1
\]
which is a contradiction. If $k_1 < k$ we take the $k$ influence sets from above together with $S_{p_L}(r)$ and $S_{p_{L'}}(r)$ again yielding a contradiction.
So the theorem is proven.
\end{proof}

But not only trees with a few number of leaves admit linear-time data-dissemination (see Theorem~\ref{thm:kleaves}), also trees  with only a few inner nodes do:
\begin{theorem}
 For trees with $\ell$ leaves we have 
$B^{\mathfrak{T}_{\ell}}\leq (n-\ell)(n-1)+2-\max(n,2(n-\ell))$.
In particular, in trees with only $k$ inner nodes (i.e., $(n-k)$ leaves), data-dissemination
 is linear. In fact, $B^{\mathfrak{T}_{n-k}}\leq k(n-1)+2-\max(n,2k)$.
\end{theorem}
\begin{proof}
 After the first round we can find a covering of size $(n-\ell)$ sets (by taking all influence sets except those of the leaves), all of size at least $2$. 
 By Lemma~\ref{lem:cov}~\eqref{lem:cov:onesetgrows}, in all following rounds at least of them must grow.
\end{proof}
\begin{remark}
 Note that in case of the star graph (i.e., $n-1$ leaves and $1$ inner node) this theorem indeed gives dissemination time of $1$.
\end{remark}

Turning back to general rooted trees $\mathfrak{T}$, the following theorem presents a lower
bound on the dissemination time. It reveals that, in the worst-case, it takes more time 
than in the case of chain graphs.

\begin{theorem}
For the class of rooted trees $\mathfrak{T}$, $B^\mathfrak{T} \geq 
\lceil \tfrac{3n-1}{2} \rceil-2$ rounds.
\end{theorem}

\begin{proof}

We will construct a specific sequence $\mathcal{G}$ of graphs with $B_{\mathcal{G}} = \lceil \tfrac{3n-1}{2} \rceil-2$.
This sequence consists of three different graphs $G^{(1)},G^{(2)},G^{(3)}$ where each graph
will be applied for multiple rounds.

The first graph, $G^{(1)}$, is the simple chain rooted in the process $1$ and edges
$i \rightarrow {i+1}$. 

The tree $G^{(2)}$ is rooted in $n$ and contains edges $(n \rightarrow 1),(n
\rightarrow {n-1})$. 
Furthermore $(i \rightarrow {i+1})$ for $1\leq i \leq \lfloor
\frac{n}{2}\rfloor-1$,
and $(i \rightarrow {i-1})$ for $n-1\geq i \geq \lfloor
\frac{n}{2}\rfloor+2$.

Finally, $G^{(3)}$ is rooted in $\lfloor \frac{n}{2} \rfloor$ with edges $i \rightarrow
{i+1}$ for $\lfloor \frac{n}{2} \rfloor\leq i \leq n-1$ and 
edges $i \rightarrow {i+1}$ for $1 \leq i \leq \lfloor \frac{n}{2}
\rfloor-1$.
Furthermore there is an edge $n \rightarrow 1$.

The execution (\cref{fig:counter}) is constructed in the following way:
\begin{itemize}
\item $G_r=G^{(1)}$ for $1\leq r \leq \lfloor\frac{n-1}{2} \rfloor$,
\item $G_r=G^{(2)}$ for $\lfloor\frac{n-1}{2} \rfloor+1 \leq r \leq n-2$, and
\item $G_r=G^{(3)}$ for $n-1\leq r \leq \lceil \frac{3n-1}{2} \rceil-2$.
\end{itemize}
In this sequence, the first time an influence set has size $n$ is the last
round. Hence $\lceil \frac{3n-1}{2}\rceil-2$ is a lower bound for broadcasting in directed
trees.

\newcommand{\impscale}{0.9}
\begin{figure*}[h]
 \centering
   \begin{subfigure}{0.05\linewidth}
     \scalebox{\impscale}{
       \begin{tikzpicture}
         \node at (0, 0.5) {$\sigma_1$:};
       \end{tikzpicture}
     }
   \end{subfigure}
   \begin{subfigure}{.16\linewidth}
     \scalebox{\impscale}{
       \begin{tikzpicture}
         \node (p1) at (0,0) {$1$};
         \node (p2) at (0,-0.7) {$2$};
         \node[rotate=90] (dots) at (0,-1.5) {\dots};
         \node (pn) at (0,-2.3) {$n$};

         \draw[->] (p1) -- (p2);
         \draw[->] (p2) -- (dots);
         \draw[<-] (pn) -- (dots);

         \node (left-paren)  at ( -.6, -1.1) {$\left( \vphantom{\rule{1pt}{44pt}} \right.$};
         \node (right-paren) at (1,-1.05) {$\left. \vphantom{\rule{1pt}{44pt}}
         \right)_{1}^{\lfloor\frac{n-1}{2}\rfloor}$};
       \end{tikzpicture}
     }
   \end{subfigure}
   \begin{subfigure}{.35\linewidth}
     \scalebox{\impscale}{
       \begin{tikzpicture}
         \node (pn) at (0,0) {$n$};
         \node (p1) at (-.8,-.8) {$1$};
         \node (pn-1) at (0.8,-.8) {${n-1}$};
         \node[rotate=90] (dots) at (0.8,-1.6) {\dots};
         \node[rotate=90] (dots2) at (-.8,-1.6) {\dots};
         \node (pn2) at (-.8,-2.5) {${\lfloor \frac{n}{2} \rfloor}$};
         \node (pn2+1) at (0.8,-2.5) {${\lfloor \frac{n}{2} \rfloor+1}$};

         \draw[->] (pn) -- (p1);
         \draw[->] (pn) -- (pn-1);
         \draw[->] (pn-1) -- (dots);
         \draw[->] (p1) -- (dots2);
         \draw[->] (dots) -- (pn2+1);
         \draw[->] (dots2) -- (pn2);

         \node (left-paren)  at ( -1.5, -1.2) {$\left( \vphantom{\rule{1pt}{44pt}} \right.$};
         \node (right-paren) at (2.3,-1.15) {$\left. \vphantom{\rule{1pt}{44pt}}
         \right)_{\lfloor\frac{n-1}{2}\rfloor+1}^{n-2}$};
       \end{tikzpicture}
     }
   \end{subfigure}
   \begin{subfigure}{.26\linewidth}
     \scalebox{\impscale}{
       \begin{tikzpicture}
         \node (p1) at (0,0) {${\lfloor \frac{n}{2} \rfloor}$};
         \node[scale=0.9] (p2) at (0,-0.7) {${\lfloor \frac{n}{2} \rfloor+1}$};
         \node[rotate=90] (dots) at (0,-1.5) {\dots};
         \node (pn) at (0,-2.3) {${\lfloor \frac{n}{2} \rfloor-1}$};

         \draw[->] (p1) -- (p2);
         \draw[->] (p2) -- (dots);
         \draw[<-] (pn) -- (dots);

         \node (left-paren)  at ( -0.9, -1.1) {$\left( \vphantom{\rule{1pt}{44pt}} \right.$};
         \node (right-paren) at (1.4,-1.05) {$\left. \vphantom{\rule{1pt}{44pt}}
         \right)_{n-1}^{\lceil\frac{3n-1}{2}\rceil-2}$};
       \end{tikzpicture}
     }
   \end{subfigure}\\
    \begin{subfigure}{0.05\linewidth}
     \scalebox{\impscale}{
       \begin{tikzpicture}
       \end{tikzpicture}
     }
   \end{subfigure}
   \begin{subfigure}{.35\linewidth}
     \scalebox{\impscale}{
       \begin{tikzpicture}
       \end{tikzpicture}
     }
   \end{subfigure}
   \begin{subfigure}{.26\linewidth}
     \scalebox{\impscale}{
       \begin{tikzpicture}
       \end{tikzpicture}
     }
   \end{subfigure} \caption{example execution
 }\label{fig:counter}
\end{figure*}
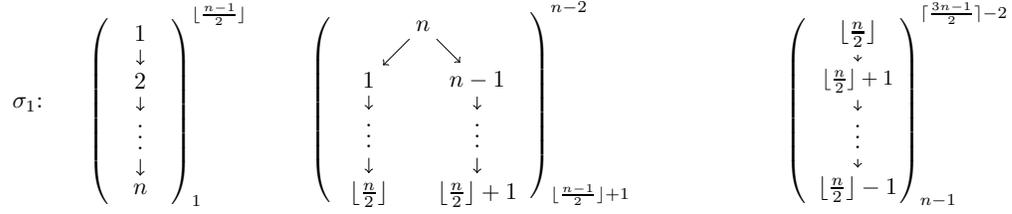
\begin{strip}
\begin{align*}
\texttt{for~} 1 \leq r\leq\lfloor \frac{n-1}{2}\rfloor:&S_i(r)=\{i,\dots,\min(r+i,n) \}\\
\texttt{for~} \lfloor\frac{n-1}{2}\rfloor<r\leq n-2:
&S_i(r)=\{i,\dots,\lfloor\frac{n-1}{2}\rfloor+i
\}\texttt{~for~}i\leq\frac{n}{2},\\
&S_i(r)=\{\max(\lfloor\frac{n}{2}\rfloor+1,i-(r-\lfloor\frac{n-1}{2}\rfloor)),\dots,
n,1,\dots,r-\lfloor\frac{n-1}{2}\rfloor \}\text{~for~}i>\frac{n}{2}\\
\texttt{for~} n-2<r\leq\lceil \frac{3n-1}{2}\rceil-2:
&S_i(r)=\{\max(\lfloor\frac{n}{2}\rfloor+1,i+2+\lfloor\frac{-n-1}{2}\rfloor)),\dots,
n,1,\dots,\lceil\frac{n+1}{2}\rceil-2 \}\texttt{~for~}i>\frac{n}{2},\\
&S_i(r)=\{i,\dots,\lfloor\frac{n-1}{2}\rfloor+i+r-(n-2) \}\texttt{~for~}i\leq
\lceil\frac{3n+1}{2}\rceil-2-r,\\
&S_i(r)=\{i,\dots,n,1,\dots, r-(n-2)\}\texttt{~for~}
\frac{n}{2}\geq i>\lceil\frac{3n+1}{2}\rceil-2-r)
\end{align*}
\end{strip}
\end{proof}

\section{Undirected Trees}\label{sec:undir}

In undirected graphs, dropping the direction of the
edges speeds-up data dissemination.
In the case where $\mathfrak{G}$ is the class of undirected and connected graphs, after $n-1$ rounds, even
all-to-all dissemination is completed (see \cite{KLO10}[Proposition 3.1]).
But what can be said about the dissemination time?
The following theorem shows that dissemination is twice as
fast as all-to-all dissemination in undirected chains, and 
roughly at least $3/2$ as fast as in directed rooted trees.

\begin{theorem}
Let $\mathfrak{C}_u$ be the class of undirected chains. At the end of 
round $r$, there exists a set $\mathcal{S}(r)$ of $n-2r$ influence
sets $S_{p_1}(r),\dots ,S_{p_{n-2r}}(r)$ with 
\begin{itemize}
 \item $|S_{p_i}| \geq 2r+1$ for $1\leq i \leq n-2r$, and
 \item for all $0 \leq r \leq (n-1)/2$
\begin{equation}\label{equ:unionsize_undir}
 \left| \bigcup_{i \in I} S_{p_i}(r) \right| \geq 2r+1 + |I|-1 \quad \forall \ I \subseteq [1,n-2r] .
\end{equation}
\end{itemize}
Thus, $B^{\mathfrak{C}_u} \leq \lceil (n-1)/2 \rceil$ rounds.
\end{theorem}
Note that this bound is also tight, as the constant chain graph reveals.

\begin{proof}
The proof runs along the same lines as the proof for the analogous result for rooted chains and uses induction on $r$. For $r=0$ it is obvious true. 

Now let's do the induction step: Again, we take successively a set $S_{p_i}(r+1)$ (where $S_{p_i}(r)$ was contained in $\mathcal{S}(r)$) and add it to 
$\mathcal{S}(r+1)$ iff the following two conditions hold:
\begin{itemize}
\item[(i)] $\left| S_{p_i}(r+1) \right| \geq 2r+3$ and
\item[(ii)] inequality~\eqref{equ:unionsize_undir} holds for all collection of sets of $\mathcal{S}(r+1)$ and $S_{p_i}(r+1)$.
\end{itemize}

Here condition~(i) holds for all but at most two sets from $\mathcal{S}(r)$: A set grows by exactly $1$ iff it is located at one of the ends of the chain. Since
all sets of size $(2r+1)$ are pairwise different at most two such sets grow by at most $1$. If such sets exists we denote them by $S_{p_i^*}$, $i \in \{1,2\}$.

Assume now that we have already added successfully $(L-1)$ sets to $\mathcal{S}(r+1)$ and that by adding $S_{p_L}(r+1)$ condition (ii) is violated, i.e., there exists
$k \leq (L-1)$ sets $S_{p_{i_1}}(r+1),\dots,S_{p_{i_k}}(r+1) \in \mathcal{S}(r+1)$ such that
\[
\left| \bigcup_{\ell = 1}^k S_{p_{i_{\ell}}}(r+1) \cup S_{p_L}(r+1) \right| \leq 2r+k+2.
\]

By induction hypothesis,
\[
\left| \bigcup_{\ell = 1}^k S_{p_{i_{\ell}}}(r) \cup S_{p_L}(r) \right| \geq 2r+1+k.
\]
Thus the set $\bigcup_{\ell = 1}^k S_{p_{i_{\ell}}}(r) \cup S_{p_L}(r)$ grew only by 1 in round $(r+1)$, or equivalently, it captures the last $(2r+1+k)$ elements of one end of
the chain in round $(r+1)$. 

Firstly, in case of the existence of $S_{p_i^*}$ at the same end, also this set contains exactly the last $(2r+1)$ elements, which gives
\[
\left| \bigcup_{\ell = 1}^k S_{p_{i_{\ell}}}(r) \cup S_{p_L}(r) \cup S_{p_i^*}(r) \right| = 2r+1+k
\]
which is a contradiction to the induction hypothesis. 
Hence -- if two sets $S_{p_i^*}$ exists -- we can add all influence sets of $\mathcal{S}(r)$ but $S_{p_i^*}$ to $\mathcal{S}(r+1)$ and
the assertion of this theorem holds in this case. 

Secondly, if exactly one set $S_{p*}$ exists (at the opposite end) we are allowed to delete one set from our sets fulfilling condition~(ii), so we do not add $S_{p_L}(r+1)$ to $\mathcal{S}(r+1)$.
If there would be an index $L'>L$ such that again condition (ii) is violated, there would be also $k_1<L'$ sets
$S_{p_{i_1}}(r+1),\dots,S_{p_{i_{k_1}}}(r+1)$ with
\[
\left| \bigcup_{\ell = 1}^{k_1} S_{p_{i_{\ell}}}(r+1) \cup S_{p_{L'}}(r+1) \right| = 2r+3+k_1.
\]
Again, the set $\bigcup_{\ell = 1}^{k_1} S_{p_{i_{\ell}}}(r) \cup S_{p_{L'}}(r)$ captures the last $(2r+1+k_1)$ elements of the same end of the chain
in round $(r+1)$. If $k_1\geq k$ then we take the $k_1$ influence sets from here together with $S_{p_L}(r)$ and $S_{p_{L'}}(r)$ and obtain
\[
\left| \bigcup_{\ell = 1}^{k_1} S_{p_{i_{\ell}}}(r) \cup S_{p_{L'}}(r) \cup S_{p_{L}}(r) \right| = 2r+1+k_1
\]
which is a contradiction. If $k_1 < k$ we take the $k$ influence sets from above together with $S_{p_L}(r)$ and $S_{p_{L'}}(r)$ again yielding a contradiction.

Thirdly, In case of absence of sets $S_{p*}$ we are allowed to delete two sets from $\mathcal{S}(r)$. Since by the argument above at one end at most one set violates condition (ii) we are done.
So the theorem is proven.
\end{proof}

\section{Acknowledgments}
We would like to thank our colleague Kyrill Winkler for several useful discussions on this topic.

\section{Outlook}
\label{sec:outlook}
We presented a number of lower and upper bounds for the dissemination time
in dynamic networks under oblivious message adversaries, where the set of
admissible graphs is restricted to contain trees or fixed substructures 
of trees.
For rooted directed trees, the best upper bound is $O(nlog(n))$ and the
best lower bound is $\Omega(n)$.
Hence, it is still an open question whether the dissemination time is
indeed linear or not.
Besides our interest in finally closing this question, we are wondering
whether the worst case dissemination time is somehow connected to the 
maximum path length, diameter or the maximum node degree of the individual
communication graphs or their dynamic transitive closure in the graph
sequence. Moreover, it remains to be seen whether our chain upper bound 
also holds for undirected trees.
\newpage
\bibliographystyle{plain}
\bibliography{../../paper,../../lit_bib/lit,../../lit_bib/journalstrings_abbr,../../Lit_Broadcasting_Rand_Dyn_Graphs}

\end{document}